\documentclass[11pt]{llncs}
\usepackage{graphicx}

\usepackage[matrix,frame,arrow]{xy}
\usepackage{amsmath}

\newcommand{\ket}[1]{\left\vert{#1}\right\rangle}
\newcommand{\qw}[1][-1]{\ar @{-} [0,#1]}
\newcommand{\qwx}[1][-1]{\ar @{-} [#1,0]}

\newcommand{\gate}[1]{*{\xy *+<.6em>{#1};p\save+LU;+RU **\dir{-}\restore\save+RU;+RD **\dir{-}\restore\save+RD;+LD **\dir{-}\restore\POS+LD;+LU **\dir{-}\endxy} \qw}
\newcommand{\meter}{\gate{\xy *!<0em,1.1em>h\cir<1.1em>{ur_dr},!U-<0em,.4em>;p+<.5em,.9em> **h\dir{-} \POS <-.6em,.4em> *{},<.6em,-.4em> *{} \endxy}}

\newcommand{\control}{*-=-{\bullet}}

\newcommand{\ctrl}[1]{\control \qwx[#1] \qw}

\newcommand{\multigate}[2]{*+<1em,.9em>{\hphantom{#2}} \qw \POS[0,0].[#1,0];p !C *{#2},p \save+LU;+RU **\dir{-}\restore\save+RU;+RD **\dir{-}\restore\save+RD;+LD **\dir{-}\restore\save+LD;+LU **\dir{-}\restore}
\newcommand{\ghost}[1]{*+<1em,.9em>{\hphantom{#1}} \qw}

\newcommand{\lstick}[1]{*!R!<.5em,0em>=<0em>{#1}}

\newcommand{\Qcircuit}{\xymatrix @*=<0em>}

\newcommand{\polylog}{\operatorname{polylog}}

\begin{document}

\title{Quantum Online Memory Checking\thanks{This article has appeared as: ``Quantum Online Memory Checking'', 
Wim van Dam and Qingqing Yuan. In \emph{Theory of Quantum Computation, Communication, and Cryptography: Fourth Workshop, TQC 2009, Waterloo, Canada, May 11--13, 2009. Revised Selected Papers,} eds.\ Andrew Childs and Michele Mosca, Lecture Notes in Computer Science, Volume 5906, Springer, pages 10--19 (2009)}}
\author{Wim van Dam\inst{1,2}
\and
Qingqing Yuan\inst{1,3}}
\institute{Department of Computer Science, University of California, Santa Barbara \and
Department of Physics, University of California, Santa Barbara
\and 
Microsoft Corporation, Redmond, WA}
\maketitle

\begin{abstract}
  The problem of memory checking considers storing files on an
  unreliable public server whose memory can be modified by a
  malicious party. The main task is to design an online memory checker
  with the capability to verify that the information on the server has
  not been corrupted. To store $n$ bits of public information, the
  memory checker has $s$ private reliable bits for verification
  purpose; while to retrieve each bit of public information the
  checker communicates $t$ bits with the public memory. Earlier work
  showed that, for classical memory checkers, the lower bound $s\times
  t\in\Omega(n)$ holds. In this article we study \emph{quantum} memory
  checkers that have $s$ private qubits and that are allowed to
  quantum query the public memory using $t$ qubits. We prove an
  exponential improvement over the classical setting by showing the
  existence of a quantum checker that, using quantum fingerprints,
  requires only $s\in O(\log n)$ qubits of local memory and $t\in
  O(\polylog n)$ qubits of communication with the public memory.
\end{abstract}

\section{Introduction}\label{sec:intro}
The problem of memory checking was first introduced by Blum et al.\ \cite{BlumEvansGemmell:1994} as an extension of program checking. In this
problem, a memory checker receives a sequence of ``store'' and
``retrieve'' operations from a user, and the checker has to relay
these commands to an unreliable server. By making additional
requests to the unreliable memory and using a small private and
reliable memory for storing additional information, the checker is
required to give correct answers (with high probability) to the
user's retrieve operations that are in accordance the previous
store instructions, or report error when the information has been
corrupted. Blum et al.\ \cite{BlumEvansGemmell:1994} made a distinction
between ``online'' and ``offline'' memory checkers: an
\emph{online memory checker} must detect the error immediately after receiving an
errant response from the memory, while an \emph{offline checker}
is allowed to output whether they are all handled correctly until
the end of the operation sequence.

There are two main complexity measures regarding memory checkers: the
\emph{space complexity}, which is the size of its private reliable
memory, and the \emph{query complexity}, which is the size of the
messages between the memory checker and the public memory per user
request. The goal is to have a reliable checker with low space
complexity and low query complexity against any (probabilistic,
polynomial time) adversary corrupting the public memory.

With $s$ be the space complexity and $t$  the query complexity of
an online memory checker, both Blum et al.\ \cite{BlumEvansGemmell:1994}
and Naor and Rothblum~\cite{NaorRothblum:2009} proved that for classical online
memory checking one has the lower bound $s\times t\in\Omega(n)$.

\subsubsection{Our Result}

We looked at the efficiency of online memory checkers that are allowed to operate in a quantum mechanical way. After defining the proper model, we present an online memory checker using quantum fingerprints that requires only $s\in O(\log n)$ bits of private memory and $t\in O(\polylog n)$ queries to the public memory. We also prove its correctness and security. Specifically we show that for an error rate $\epsilon>0$ it is sufficient for the memory checker to privately keep $O(\log(1/\epsilon))$ copies of the quantum fingerprints of the public memory (each requiring $O(\log n)$ qubits). The parameters of the specific error correcting code that we use for the quantum fingerprints introduces a constant multiplicative term in this quantity $O(\log (1/\epsilon))$.

\subsubsection{Acknowledgment}
This material is based upon work supported by the National Science Foundation under Grant No.\ 0729172 (``Quantum Algorithms for Data Streams'').

\section{Preliminaries}\label{sec:prelim}
In this section, we present the model of memory checker in the
quantum settings. We also briefly review some of the techniques
used in our quantum algorithms for online memory checking.
\subsection{Memory Checker}
We first introduce the classical definition of memory checker, and then extend it to quantum settings.

\begin{definition} \textbf{Classical Memory Checkers (see
   \cite{BlumEvansGemmell:1994,NaorRothblum:2009}).}  A memory checker is a probabilistic Turing
 machine $\mathcal{C}$ with five tapes: a read-only input tape for
 $\mathcal{C}$ to read the requests from user $\mathcal{U}$, a
 write-only output tape for $\mathcal{C}$ to write its response to
 the user's requests or that the memory $\mathcal{M}$ is
 ``{buggy}'', a write-only tape for $\mathcal{C}$ to write
 requests to $\mathcal{M}$, a read-only tape for $\mathcal{C}$
 reading response from $\mathcal{M}$, and a read-write work tape as a
 secret, reliable memory.
\end{definition}
\paragraph{\textbf{Quantum Memory checker:}}
In our quantum mechanical extension of this definition, the input
and output tape between $\mathcal{C}$ and $\mathcal{U}$ both
remain classical, as well as the memory $\mathcal{M}$.  The
checker $\mathcal{C}$, however, is now allowed to make quantum
queries to the memory $\mathcal{M}$ and the secret work-tape of
$\mathcal{C}$ and the two read and write-only tapes between
$\mathcal{C}$ and $\mathcal{M}$ now support quantum bits.  This
model is illustrated in Fig.~\ref{checker}.
\begin{figure}[htb]
\includegraphics[height = 5.5cm]{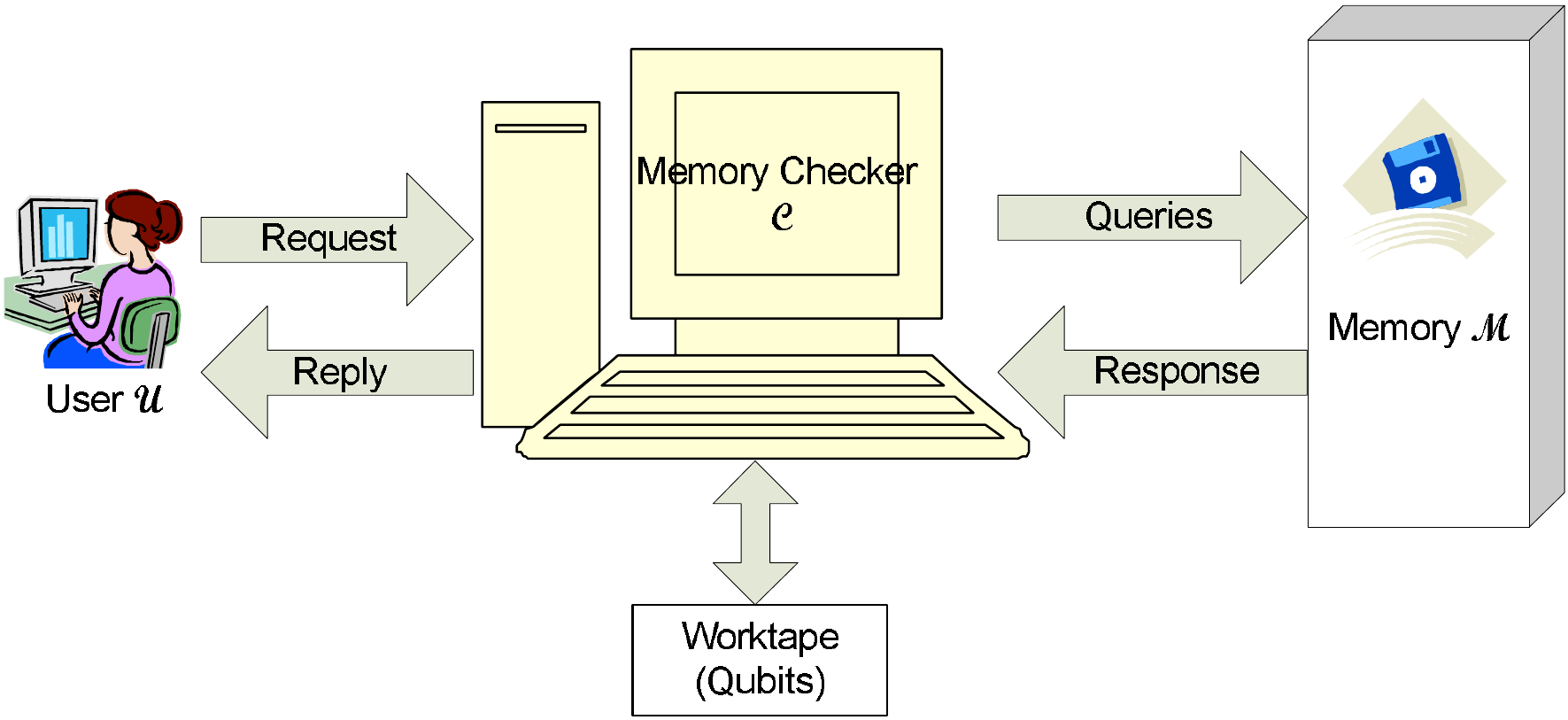} 
\caption{A quantum mechanical memory checker: The user presents
  classical ``store'' or ``retrieve'' request to the checker, which,
  with high probability, returns the correct answer or reports
  ``buggy'' when the memory has been corrupted. The checker can make quantum queries to the
  memory, such that it acquires a superposition of values.  In
  addition, the checker is also allowed to have a private, secure
  work tape that consists of qubits and that is much smaller than
  the public memory. } \label{checker}
\end{figure}

The user $\mathcal{U}$ presents the ``store'' and ``retrieve''
requests to $\mathcal{C}$ and after each ``retrieve'' request,
$\mathcal{C}$ must write an answer to the output tape or output
that $\mathcal{M}$ is ``buggy'' if the public memory
$\mathcal{M}$ has been corrupted. We say a memory \emph{acts
correctly} if the returns of a ``retrieve'' operation are
consistent with the contents written by the previous adjacent
``store'' operation. For any operation sequence of polynomial
length in the total size $n$ of the data stored by $\mathcal{U}$
on $\mathcal{M}$ and error rate $0<\epsilon<\frac{1}{2}$, it is
required that:
\begin{itemize}
\item If $\mathcal{M}$'s output to the ``retrieve'' operation is
correct, $\mathcal{C}$ also answers $\mathcal{U}$'s request with
correctness probability at least $1-\epsilon$. 
\item If $\mathcal{M}$'s output is incorrect for some operation,
  $\mathcal{C}$ outputs ``buggy'' with probability at least
  $1-\epsilon$.
\end{itemize}

There are two important measures of the complexity of a memory
checker: the size $s$ of its secret memory (the \emph{space
 complexity}) and the number $t$ of bits exchanged between
$\mathcal{C}$ and $\mathcal{M}$ per request from the user (the
\emph{query complexity}). We follow the convention that we only
consider the \emph{query complexity} for \emph{retrieve requests}
such that the query complexity for \emph{store requests} may be
unbounded. Obviously, if the secret memory is sufficiently large,
the solution to this problem is trivial as $\mathcal{C}$ can
simply store the $n$ bits on its work-tape. More interesting is
the case where the space complexity $t$ is sublinear (typically
logarithmic) in $n$.

As noted in \cite{BlumEvansGemmell:1994} and~\cite{NaorRothblum:2009}, memory checkers can be
categorized into ``online'' and ``offline'' versions. In the
offline model, the checker $\mathcal{C}$ is allowed to output
``buggy'' at any point before the last ``retrieve''
request in the sequence if $\mathcal{M}$'s answers to some request
is incorrect. The online model is more restricted as $\mathcal{C}$
is required to detect the error immediately once $\mathcal{M}$
gives an incorrect answer to the request.

In this paper, we focus on online memory checkers. As noted in the
Introduction, it is known that for classical online memory
checkers, we have the lower bound $s\times
t\in\Omega(n)$~\cite{NaorRothblum:2009}. Below it will be shown that with
quantum memory checkers one can get an exponential reduction on
this lower bound.

\subsection{Quantum Simultaneous Message Protocol}\label{finger}
Buhrman et al.\ \cite{BuhrmanCleveWatrous:2001} extended the classical
simultaneous message (SM) model~\cite{Yao:1979,NewmanSzegedy:1996} to the quantum
setting. In this model there are three players: Alice has a
bit-string $x$, Bob has another bit-string $y$, and they do not
share entanglement or randomness, but they each send one quantum
message to a referee, Carol, who tries to compute the function
value $f(x,y)$. The complexity measure of this protocol is the
number of qubits used in the messages. Classically, for the
\textsc{Equality} function (Carol's output has to be $f=1$ if
$x=y$, and $f=0$ otherwise), Newman and Szegedy~\cite{NewmanSzegedy:1996} showed
that the randomized SM complexity of the \textsc{Equality}
function on $\{0,1\}^n$ has the lower bound $\Omega(\sqrt n)$.
Buhrman et al. presented a quantum protocol for the
\textsc{Equality} function that enabled the referee to compute
$f(x,y)$ by comparing the two ``quantum fingerprints''
$\ket{\psi_x}$ and $\ket{\psi_y}$ of $x$ and $y$ sent by Alice and
Bob, respectively. The communication complexity of this protocol
is $O(\polylog n)$ qubits.

The protocol works as following: for $x, y\in \{0,1\}^n$ we use an
error correcting code $E:\{0,1\}^n\rightarrow\{0,1\}^m$ with
$m=cn$. The Hamming distance between two distinct codewords $E(x)$
and $E(y)$ (with $x\neq y$) is at least $\delta m$, with
$\delta>0$ a constant. Let $E_i(x)$ denote the $i^{th}$ bit of
$E(x)$. Alice constructs the superposition
\begin{equation*}
\ket{\psi_x}=\frac{1}{\sqrt m}\sum_{i=1}^m(-1)^{E_i(x)}\ket{i}
\end{equation*}
as the fingerprint of her input $x$. Similarly Bob construct
$\ket{\psi_y}$ for his input $y$ and both of them send the
fingerprints (of size $\log m = O(\log n)$ qubits) to Carol. Carol
performs the ``Controlled-SWAP test'' shown in the following
circuit:
\[ \Qcircuit @C=1em @R=.7em {
 \lstick{\ket{0}} & \gate{H} & \ctrl{1} & \gate{H} & \meter  \\
 \lstick{\ket{\psi_x}} & \qw & \multigate{1}{\mathrm{SWAP}} & \qw & \qw \\
 \lstick{\ket{\psi_y}} & \qw & \ghost{\mathrm{SWAP}} & \qw & \qw  \\
  }
\]
If the measurement of the first register is $0$, Carol decides
that $x=y$; otherwise she concludes $x\neq y$. It is easy to show
that the probability of Carol measuring ``0'' equals
$\frac{1}{2}+\frac{1}{2}|\langle\psi_x|\psi_y\rangle|^2$ and the
probability of measuring ``1'' is
$\frac{1}{2}-\frac{1}{2}|\langle\psi_x|\psi_y\rangle|^2$.
Therefore, when $x=y$ the probability of Carol measuring ``0'' is
$1$, while when $x\neq y$ the probability of measuring ``$0$'' is
at most $\frac{1}{2}+\frac{1}{2}|\langle\psi_x|\psi_y\rangle|^2$.
If we perform this test repeatedly for $k$ copies of
$\ket{\psi_x}$ and $\ket{\psi_y}$ with $x\neq y$, the probability
of measuring all zeros is
$(\frac{1+|\langle\psi_x|\psi_y\rangle|^2}{2})^k$, which decays
exponentially in $k$.

\subsection{Locally Decodable Codes}
To construct the quantum fingerprints, we first encode the string
using error correcting codes. In this paper, we use a locally
decodable codes (see for example Katz and Trevisan~\cite{KatzTrevisan:2000}) such
that a single bit $x_j$ of the original data can be probabilistically
reconstructed by reading only a small number of locations in the
encoding $E(x)$. Formally speaking~\cite{KatzTrevisan:2000}, for fixed $\delta$,
$\epsilon>0$ and integer $q$ we say that $E:
\{0,1\}^n\rightarrow\{0,1\}^m$ is a
$(q,\delta,\epsilon)$-\emph{locally decodable code} (LDC) if there
exists a probabilistic algorithm that reads at most $q$ bits of $E(x)$
to determine one of the bits of $x_j$ and if that same algorithm
returns the correct value with probability at least $1/2+\epsilon$ on
all strings $y\in\{0,1\}^m$ with Hamming distance $d(y,E(x))\leq
\delta m$.

It is not important for us to choose a perfect LDC in our memory
checker. In our algorithm, considering the fact that it takes too
much time if it starts from the original string to construct the
quantum fingerprints, we encode the string and store its codeword
on the public memory to speed up the processing. On the other
hand, if we use any other error correcting code where decoding
requires to query the whole codeword, it takes too much time for
the user to retrieve a bit. This leads us to use LDCs. 
For our purposes it will be sufficient to use the construction of 
Babai et al.~\cite{BabaiFortnowLevin:1991}, who constructed an LDC with 
$q\in\polylog(n)$ queries and $m\in O(n^2)$ for fixed $\delta$
and $\epsilon$.

\section{Quantum Algorithm for Online Memory
Checking}\label{sec:alg}

In this section, we state the main theorem of this paper.

\begin{theorem}
  For any error rate $\epsilon>0$, there exists a quantum online
  memory checker with space complexity $s\in O(\log(1/\epsilon)\log
  n)$ and query complexity $t\in O(\log(1/\epsilon)\log n+\polylog n)$, where $n$
  is the size of the public bitstring. This checker answers the user
  correctly with constant probability at least $1-\epsilon$ when the memory
  $\mathcal{M}$ acts correctly, and it replies ``buggy'' with
  probability at least $1-\epsilon$ when $\mathcal{M}$ has been
  corrupted.\label{maintheorem}
\end{theorem}
We prove Theorem~\ref{maintheorem} by presenting a quantum online
memory checker with the claimed upper bounds on the space
complexity $s$ and query complexity $t$ of the checker.

\subsection{Online Memory Checking Using Quantum Fingerprints} 
The proposed quantum memory checker $\mathcal{C}$ uses the following ingredients. 
Let $x=x_1\dots x_n$ be the string that the user $\mathcal{U}$ wants
to write to the public memory $\mathcal{M}$.  
\begin{description}
\item[Public:] The memory checker $\mathcal{C}$ uses a $q$-query
  locally decodable code $E:\{0,1\}^n\rightarrow\{0,1\}^m$ and writes
  the codeword $E(x)\in\{0,1\}^m$ to the public memory $\mathcal{M}$.
\item[Private:] The memory checker maintains $k$ copies of the quantum
  fingerprint
\begin{equation*}
\ket{\psi_x} := \frac{1}{\sqrt{m}}\sum_{j=1}^m(-1)^{E(x)_j}\ket{j}
\end{equation*}
of $x$ in its private memory (the value of $k$ will be determined
later).
\end{description}
Every time a ``retrieve'' instruction is executed, the memory checker
obtains $k$ summary states $\ket{y}$ of the current state of
the public memory $\mathcal{M}$. By comparing these new quantum
fingerprints with those in the checker's private memory, the checker
can detect any malicious changes that would corrupt the decoding of
$E(x)$ to the public memory with high probability. Specifically, the
checker uses the following two protocols.
\subsubsection{Retrieve ($x_i$) protocol:}
\begin{itemize}
\item When a ``retrieve'' request is issued by the user the memory
 checker queries the public memory to obtain $k$ ``summary states''
\begin{equation*}
\ket{y} = \frac{1}{\sqrt{m}}\sum_j (-1)^{y_j}\ket{j}.
\end{equation*} 
\item The checker performs the Controlled-SWAP test on the $k$ copies
  of $\ket{y}$ and $\ket{\psi_x}$ as defined in
  Section~\ref{finger}. 
\item If any of the $k$ measurement outputs $1$, the checker replies
  ``buggy''.
\item Otherwise, the checker runs the decoding algorithm of the
  locally decodable code $E$ to reconstruct the bit $x_j$ the user
  requests (which requires $q$ queries to the public
  memory) and returns this bit to the user.
\item The checker then replaces the $\ket{\psi_x}$ fingerprints in
  its local memory with $k$ new summaries $\ket{y}$ of the public
  memory.
\end{itemize}

\subsubsection{Store ($x$) Protocol:}
\begin{itemize}
\item When a ``store'' request is issued, the checker first queries
  the public memory as in the first $3$ steps of the previous protocol
  to verify that the public memory and private fingerprints coincide
  with each other.
\item
The checker computes the codeword $E(x)$ for the new
input and writes it to the memory. 
\item It also computes new fingerprint $\ket{\psi_x}$ and stores $k$
  copies into its private memory.
\end{itemize}
The complexity measure of this protocol is as follows. For simplicity,
we assume here the sub-optimal parameters of the LDC of Babai et
al.~\cite{BabaiFortnowLevin:1991} with $q\in\polylog n$ and $m\in O(n^2)$.  The space
complexity is the private memory holding the fingerprints of $x$,
which is $O(k\log n)$ qubits; the query complexity is the number of
qubits answered by $\mathcal{M}$ per request, which includes the $k$
copies of the fingerprints and the queries of LDC; this amounts to 
$O(k\log n+\polylog n)$ qubits.

\subsection{Correctness of the Quantum Online Memory Checker}
Based on the definition of online memory checker in
Section~\ref{sec:prelim}, a correct checker should answer the user
correctly when the public memory $\mathcal{M}$ is correct with
probability at least $1-\epsilon$; and the checker should detect
the error when $\mathcal{M}$'s output is incorrect with
probability also at least $1-\epsilon$, such that
$0<\epsilon<\frac{1}{2}$ is the error rate of the protocol.  Let
us examine the behavior of our quantum online memory checker.
\begin{itemize}
\item When $\mathcal{M}$ is uncorrupted, i.e.\ when $y=E(x)$, we
have $|\langle\psi_x|y\rangle|=1$ and the probability of measuring
$0$ after the Controlled-SWAP test is $1$.  Hence the checker will
output the correct answer in this case.

\item When $\mathcal{M}$ has been changed by the adversary, i.e.\
when $y\neq E(x)$, Lemma~\ref{lemma_one} and Lemma~\ref{Tone}
applies.
\end{itemize}

\begin{lemma}
  Assume a memory checker uses error correcting codes of length
  $m$ with Hamming distance between two distinct codeword being at
  least $\delta m$ (where $\delta>0$ is a constant). With
  $k = \big\lceil\frac{\log\epsilon}{\log(1-2\delta+2\delta^2)}\big\rceil$
  copies of the fingerprint $\ket{\psi_x}$, the checker will detect
  the difference between the two fingerprints $\ket{\psi_{\tilde{x}}}$
  and $\ket{\psi_x}$ with probability at least
  $1-\epsilon$.\label{lemma_one}
\end{lemma}

\begin{proof}
Since we are using error correcting code where two distinct
codewords have Hamming distance at least $\delta m$, at least
$\delta m$ bits of the public memory have been changed. Hence for
two distinct codeword $E(x)$ and $E(\tilde{x})$,
$|\langle\psi_x|\psi_{\tilde{x}}\rangle|\leq 1-2\delta$.
Therefore, for $k$ copies, we measure all zeros with probability
at most
\begin{equation*}
\Big(\frac{1+|\langle\psi_x|\psi_{\tilde{x}}\rangle|^2}{2}\Big)^k
\leq (1-2\delta+2\delta^2)^k.\label{kcopies}
\end{equation*}
In order for the checker to detect the error of the memory with
probability at least $1-\epsilon$, the above equation should have
a value less than $\epsilon$. Therefore, if we pick
$k\geq\big\lceil\frac{\log\epsilon}{\log(1-2\delta+2\delta^2)}\big\rceil$,
the checker will output ``buggy'' with probability at
least $(1-\epsilon)$ when $\mathcal{M}$ is corrupted.
\end{proof}

Lemma~\ref{lemma_one} only deals with the situation where the
codeword is changed to another codeword. There remains one problem
though. The adversary can change a few bits of $\mathcal{M}$ in
small steps such that at no point there will be big difference
between the summary of the public memory and the private
fingerprints of the checker.  But after a sequence of such
changes, the codeword can eventually be changed into another
$E(\tilde{x})$ with $\tilde{x}\neq x$. In this situation, we have
to determine if it possible for the checker to detect the attack
with high probability. Let us formalize this situation.

\paragraph{Problem of incremental changes of public memory:} The
adversary changes a codeword $E(x)$ into another legal codeword
$E(\tilde{x})$ with $x\neq \tilde{x}$ in $T$ steps: in each step,
the adversary flips $d_i$ bits of the public memory $(1\leq i\leq
T)$, so that at step $T$, it will be changed into another
codeword, i.e.\ $\sum_{i=1}^Td_i\geq \delta m$. Without loss of
generality, we assume that in each step the adversary changes
different bits, so that once a bit is flipped in one step, it will
not be flipped back in the following steps. The problem we are
interested in is what the probability is for the checker to detect
such an attack.

In each step, the probability for the checker to accept the
response from $\mathcal{M}$ is at most
$\frac{1}{2}+\frac{1}{2}\big(|\langle\psi_x|\psi_y\rangle|^2\big)=\frac{1}{2}+\frac{1}{2}\big(1-\frac{2d_i}{m}\big)^2$.
Define $\Delta_i:=\frac{d_i}{m}$ such that
$\Delta=\sum_{i=1}^T\Delta_i\geq \delta$. Therefore, the
probability $P_T$ for the checker to measure all ``0'' (accept)
for all $T$ steps is
\begin{equation*}
P_T(\Delta_1,\dots,\Delta_T)=\prod_{i=1}^T(1-2\Delta_i+2\Delta_i^2).
\end{equation*}

\begin{lemma}
If the adversary changes $\Delta m$ bits of the codeword in $T$
steps, then the highest possible probability of the checker not
detecting the corruption is achieved if all bits get flipped in
one step. That is, for all $\Delta_i\geq 0$ with
$\Delta_1+\dots+\Delta_T=\Delta$ we have
$P_T(\Delta_1,\dots,\Delta_T)\leq P_1(\Delta)$. \label{Tone}
\end{lemma}
\begin{proof}
We prove this lemma by induction on $T$.

First, we prove that $P_2(\Delta_1,\Delta_2)\leq
P_1(\Delta_1+\Delta_2)$. We have
\begin{equation*}
P_1(\Delta_1+\Delta_2)=P_1(\Delta)=1-2\Delta+2\Delta^2
\end{equation*}
and
\begin{align*}
P_2(\Delta_1,\Delta_2) & =
P_2(\Delta_1,\Delta-\Delta_1)\\
&
=\big(1-2\Delta_1+2\Delta_1^2\big)\big(1-2(\Delta-\Delta_1)+2(\Delta-\Delta_1)^2\big)
\end{align*}
Therefore,
\begin{equation*}
P_1(\Delta_1+\Delta_2)-P_2(\Delta_1,\Delta_2)
=4\Delta_1(\Delta-\Delta_1)(\Delta+\Delta_1(\Delta-\Delta_1)) \geq
0
\end{equation*}
The last inequality holds because $0\leq\Delta_1\leq \Delta$.

Assuming the lemma holds for all $T=k-1$, let us examine $T=k$.
\begin{equation*}
P_k(\Delta_1,\dots,\Delta_k)=\prod_{i=1}^{k}\big(1-2\Delta_i+2\Delta_i^2\big)
\end{equation*}
By definition and the induction hypothesis for $T=2$ and $T=k-2$,
\begin{align*}
P_k(\Delta_1,\dots,\Delta_k)
&=P_{k-2}(\Delta_1,\dots,\Delta_{k-2})\cdot P_2(\Delta_{k-1},\Delta_k)\\
&\leq P_{1}(\Delta_1+\cdots+\Delta_{k-2})\cdot
P_1(\Delta_{k-1}+\Delta_k)\\
&\leq P_1(\Delta_1+\cdots+\Delta_k)=P_1(\Delta)
\end{align*}
Therefore, Lemma~\ref{Tone} holds for all $T\geq 1$.
\end{proof}
From this lemma it follows that the probability that the adversary
remains undetected is bounded by $P_T(\Delta_1,\dots,\Delta_T)\leq
P_1(\delta)=1-2\delta+2\delta^2$, with
$\Delta_1+\cdots+\Delta_T\geq \delta$.

The just derived probabilities are based on one copy of
$\ket{\psi_x}$ and $\ket{y}$. When we have $k$ copies, the
probability of measuring all zeros is not greater than
$(1-2\delta+2\delta^2)^k$. Therefore, if we pick $k\geq
\big\lceil\frac{\log\epsilon}{\log(1-2\delta+2\delta^2)}\big\rceil$,
the checker will output ``buggy'' with probability at
least $(1-\epsilon)$ if $\mathcal{M}$ is being corrupted.

Therefore, we can conclude that when we pick $k\geq
\big\lceil\frac{\log\epsilon}{\log(1-2\delta+2\delta^2)}\big\rceil$,
our quantum online memory checker works correctly.
Since $\delta$ and $\epsilon$ are predetermined constants, $k$ is
a constant as well. Therefore, the total complexity of this
checker is: space complexity $O(\log(1/\epsilon)\log n)$ and query
complexity $O(\log(1/\epsilon)\log n+\polylog n)$. This finishes the proof of
Theorem~\ref{maintheorem}.

Applying the same techniques as in~\cite{NaorRothblum:2009}, we have the
conclusion that our algorithm reaches the lower bound for quantum
online memory checking.

\section{Open Question}\label{sec:bounds}
The online memory checker in this article uses quantum mechanics both in its local memory and the communications with the
public memory. A variation of this model is a checker that stores
quantum information in its local memory, but communicates in
classical bits to the public memory.

In a simultaneous message protocol, if one message is quantum,
while the other is restricted to be classical, Regev and De Wolf
have shown that it requires a total of $\Omega(\sqrt{n/\log n})$
bits/qubits to compute the \textsc{Equality} function \cite{GavinskyRegevWolf:2008},
and hence such a hybrid setting is not significantly more
efficient than classical-classical protocols.  This result however
does not directly translate into a lower bound on the $s\times t$
complexity for quantum
 memory checking with classical communication.

Using the same techniques as in~\cite{NaorRothblum:2009}, a quantum online
memory checker with classical queries can be reduced to a modified
consecutive messages (CM) protocol. In this CM protocol, Alice is
allowed to send quantum messages to Carol and publish a quantum
public message, while Bob is restricted to classical messages. For
this CM protocol, there is an efficient solution as following:
Receiving an input $x$, Alice computes its quantum fingerprints
$\ket{\psi_x}$ and publish it as a public message; Bob, receiving
$y$, computes a quantum fingerprints $\ket{\psi_y}$ and compares
it with $\ket{\psi_x}$; Bob then sends Carol the result of the
Controlled-SWAP testing, who outputs the final result. The
communication complexity for this protocol is $O(\log n)$.

Due to the difference between the quantum-classical CM model and
SM protocol for \textsc{Equality} testing, it is not easy to draw
a conclusion for the lower bound of quantum online memory checking
with classical communications. Nevertheless we conjecture that
there is no efficient quantum online memory checker for this
setting.

\section{Conclusion}
In this paper, we consider the problem of constructing an online
memory checker. By using the quantum fingerprints, we reduce the
space complexity $s$ and query complexity $t$ from $s\times
t\in\Omega(n)$ to $s\in O(\log n)$ and $t\in O(\log n)$.


\begin{thebibliography}{1}

\bibitem{BabaiFortnowLevin:1991}
L.~Babai, L.~Fortnow, L.~A. Levin, and M.~Szegedy.
\newblock Checking computations in polylogarithmic time.
\newblock In {\em STOC '91: Proceedings of the twenty-third annual ACM
  Symposium on Theory Of Computing}, pages 21--32, New York, NY, USA, 1991.
  ACM.

\bibitem{BlumEvansGemmell:1994}
M.~Blum, W.~S. Evans, P.~Gemmell, S.~Kannan, and M.~Naor.
\newblock Checking the correctness of memories.
\newblock {\em Algorithmica}, 12(2/3):225--244, 1994.

\bibitem{BuhrmanCleveWatrous:2001}
H.~Buhrman, R.~Cleve, J.~Watrous, and R.~de~Wolf.
\newblock Quantum fingerprinting.
\newblock {\em Physical Review Letters}, 87:167902, 2001.

\bibitem{GavinskyRegevWolf:2008}
D.~Gavinsky, O.~Regev, and R.~de~Wolf.
\newblock Simultaneous communication protocols with quantum and classical
  messages.
\newblock {\em Chicago Journal of Theoretical Computer Science}, 2008(7),
  December 2008.
\newblock http://arxiv.org/abs/0807.2758.

\bibitem{KatzTrevisan:2000}
J.~Katz and L.~Trevisan.
\newblock On the efficiency of local decoding procedures for error-correcting
  codes.
\newblock In {\em STOC '00: Proceedings of the thirty-second annual ACM
  Symposium on Theory Of Computing}, pages 80--86, New York, NY, USA, 2000.
  ACM.

\bibitem{NaorRothblum:2009}
M.~Naor and G.~N. Rothblum.
\newblock The complexity of online memory checking.
\newblock {\em Journal of the ACM}, 56(1):1--46, 2009.

\bibitem{NewmanSzegedy:1996}
I.~Newman and M.~Szegedy.
\newblock Public vs.\ private coin flips in one round communication games
  (extended abstract).
\newblock In {\em STOC '96: Proceedings of the twenty-eighth annual ACM
  Symposium on Theory Of Computing}, pages 561--570, New York, NY, USA, 1996.
  ACM.


\bibitem{Trevisan:2004}
L.~Trevisan.
\newblock Some applications of coding theory in computational complexity.
\newblock {\em Quaderni di Matematica}, 13:347--424, 2004.


\bibitem{Yao:1979}
A.~C.-C. Yao.
\newblock Some complexity questions related to distributive computing.
\newblock In {\em STOC '79: Proceedings of the eleventh annual ACM Symposium on
  Theory of Computing}, pages 209--213, New York, NY, USA, 1979. ACM.



\end{thebibliography}
\end{document}